\documentclass[12pt,letterpaper]{article}

\usepackage{amsmath,amsfonts,amsthm,amssymb,stmaryrd,relsize}

\theoremstyle{plain}

\numberwithin{equation}{section}

\newtheorem{thm}{Theorem}[section]
\newtheorem{lem}[thm]{Lemma}
\newtheorem{cor}[thm]{Corollary}

\newenvironment{exam}[1]
{\begin{flushleft}\textbf{Example #1}.\enspace}%
{\end{flushleft}}

\allowdisplaybreaks  

\newcommand{\complex}{{\mathbb C}}
\newcommand{\Natural}{{\mathbb N}}
\newcommand{\real}{{\mathbb R}}

\newcommand{\minusone}{{\mathbin{\text{--}1}}} 

\newcommand{\underu}{\underline{u}}
\newcommand{\underv}{\underline{v}}
\newcommand{\underw}{\underline{w}}
\newcommand{\underlambda}{\underline{\lambda}}
\newcommand{\underalpha}{\underline{\alpha}}
\newcommand{\underbeta}{\underline{\beta}}
\newcommand{\overalpha}{\overline{\alpha}}

\newcommand{\rmtr}{\mathrm{tr\,}}
\newcommand{\rmsupp}{\mathrm{supp\,}}
\newcommand{\rmre}{\mathrm{Re\,}}
\newcommand{\ctimes}{\mathrel{\mathlarger\cdot}}

\newcommand{\ascript}{\mathcal{A}}
\newcommand{\bscript}{\mathcal{B}}
\newcommand{\dscript}{\mathcal{D}}
\newcommand{\lscript}{\mathcal{L}}
\newcommand{\sscript}{\mathcal{S}}

\newcommand{\ab}[1]{\left|#1\right|}
\newcommand{\doubleab}[1]{\left|\left|#1\right|\right|}
\newcommand{\brac}[1]{\left\{#1\right\}}
\newcommand{\paren}[1]{\left(#1\right)}
\newcommand{\sqbrac}[1]{\left[#1\right]}
\newcommand{\elbows}[1]{{\left\langle#1\right\rangle}}
\newcommand{\ket}[1]{{\left|#1\right>}}
\newcommand{\bra}[1]{{\left<#1\right|}}

\begin{document}

\title{A THEORY OF ENTANGLEMENT}
\author{Stan Gudder\\ Department of Mathematics\\
University of Denver\\ Denver, Colorado 80208\\
sgudder@du.edu}
\date{}
\maketitle

\begin{abstract}
This article presents the basis of a theory of entanglement. We begin with a classical theory of entangled discrete measures in Section~1. Section~2 treats quantum mechanics and discusses the statistics of bounded operators on a Hilbert space in terms of context coefficients. In Section~3 we combine the work of the first two sections to develop a general theory of entanglement for quantum states. A measure of entanglement called the entanglement number is introduced. Although this number is related to entanglement robustness, its motivation is not the same and there are some differences. The present article only involves bipartite systems and we leave the study of multipartite systems for later work.
\end{abstract}

\section{Entangled Probability Measures}  
Entangled states are considered to be an important resource for quantum computation and information  processes \cite{bus03,hz12,nc00}. Various authors have developed theories of entanglement \cite{cla05,hhhh09,pv07,ste03} and this article is another attempt. Our motivation is a bit different and we hope this work will be useful.

It is frequently stated that entanglement is a strictly quantum phenomenon and it is not present in a classical theory. We do not believe this is actually true and begin with a classical theory of entangled measures. This theory is quite simple and does not have the depth of its quantum counterpart. However, we believe that it can be instructive and give insights into the quantum theory.

Classical statistical systems are described by probability measures in a measure space. For simplicity we consider the set of probability measures $M$ on the set of natural numbers $\Natural$. We consider $\underu\in M$ as a probability vector
$\underu =\brac{u_i\colon i\in\Natural}$, $u_i\ge 0$, $\sum u_i=1$. Thinking of $M$ as a subset of the real Hilbert space
\begin{equation*}
\ell _2=\brac{f\colon\Natural\to\real ,\sum\ab{f(i)}^2<\infty}
\end{equation*}
we write $\doubleab{\underu}^2=\sum u_i^2$ and $\elbows{\underu ,\underv}=\sum u_iv_i$. The \textit{support} of $\underu$ is defined by
\begin{equation*}
\rmsupp (\underu )=\brac{i\in\Natural\colon u_i\ne 0}
\end{equation*}
The \textit{entanglement index} of $\underu$ is the cardinality of $\rmsupp (\underu )$ and is denoted by $n(\underu )$. We define the
\textit{entanglement number} of $\underu$ by $e(\underu )=\paren{1-\doubleab{\underu}^2}^{1/2}$. We can also write
\begin{equation}                
\label{eq11}
e(\underu )=\paren{\sum _{i\ne j}u_iu_j}^{1/2}=\sqbrac{\sum _iu_i(1-u_i)}^{1/2}
\end{equation}
Notice that $\doubleab{\underu}^2$ is the expectation of the random variable $u_i$ relative to the measure $\underu$ and that by \eqref{eq11}, $e(\underu )^2$ is the expectation of the random variable $1-u_i$. That is, $e(\underu )^2$ is the average distance of $\underu$ from unity. This is our first (among many) interpretations of $e(\underu )$. We say that $\underu$ is a \textit{point} (or \textit{Dirac}) measure if $u_i=1$ for some $i\in\Natural$. Of course, it follows that $u_j=0$ for $j\ne i$. We say that $\underu$ is \textit{uniform} if $u_i=u_j$ whenever $u_i,u_j\ne 0$. If $\underu$ is uniform, then $n(\underu )<\infty$ and $u_i=1/n(\underu )$ whenever $u_i\ne 0$. The proof of the following result is standard.

\begin{thm}    
\label{thm11}
{\rm{(a)}}\enspace $e(\underu )=0$ if and only if $\underu$ is a point measure.
{\rm{(b)}}\enspace If $n(\underu )<\infty$, then $e(\underu )\le\sqbrac{\paren{n(\underu)-1}/n(\underu )}^{1/2}$ and equality is achieved if and only if $\underu$ is uniform.
\end{thm}

If $\underu$ is uniform and $n(\underu )\ne 1$ (equivalently $e(\underu )\ne 0$), we say that $\underu$ is \textit{maximally entangled} with index $n(\underu )$. We conclude that there is precisely one maximally entangled probability measure for each nonsingleton finite support in
$\Natural$. Moreover, $\underu$ is maximally entangled if and only if $n(\underu )\ne 1$ and $\doubleab{\underu}^2=1/n(\underu )$. Of course, in this case $\underu$ has the largest entanglement number of any $\underv\in M$ with $n(\underv )=n(\underu )$. We also see that
$0\le e(\underu )<1$ and since $\doubleab{\underu}>0$, there is no $\underu\in M$ with $e(\underu )=1$.

\begin{exam}{1}  
(a)\enspace If $u_1=u_2=1/2$, then $e(\underu )=1/\sqrt{2\,}$ and $\underu$ is maximally entangled with index~2.
(b)\enspace If $u_1=u_2=u_3=1/3$, then $e(\underu )=\sqrt{2/3\,}$ and $\underu$ is maximally entangled with index~3 so the entanglement is larger than in (a).
(c)\enspace If $u_1=1/2$, $u_2=1/3$, $u_3=1/6$, then $e(\underu )=\sqrt{11/18\,}$ and
\begin{equation*}
\frac{1}{\sqrt{2}\,}<\sqrt{\tfrac{11}{18}\,}<\sqrt{\tfrac{2}{3}\,}
\end{equation*}
(d)\enspace If $u_1=1/9$, $u_2=1/9$, $u_3=7/9$, then $e(\underu )=\sqrt{30\,}/9<1/\sqrt{2\,}$. This gives the smallest entanglement of the four.
\hfill\qedsymbol
\end{exam}

If $\underu ,\underv\in M$ and $\lambda\in\sqbrac{0,1}$, then $\lambda\underu +(1-\lambda )\underv\in M$ is called a \textit{mixture} of $\underu$ and $\underv$. It is easy to check that 
\begin{equation*}
n\sqbrac{\lambda\underu +(1-\lambda )\underv}=n(\underu )+n(\underv )
\end{equation*}
when $\lambda\in (0,1)$, $\rmsupp (\underu )\cap\rmsupp (\underv )=\emptyset$ and
\begin{equation*}
n\sqbrac{\lambda\underu +(1-\lambda )\underv}\le n(\underu )+n(\underv )
\end{equation*}
in general. However, we have that
\begin{equation*}
n\sqbrac{\lambda\underu +(1-\lambda )\underv}\ge\lambda n(\underu )+(1-\lambda )n(\underv )
\end{equation*}
This last inequality says that the function $n$ is \textit{concave}. We interpret this as saying that mixtures increase the entanglement index. We now show that the entanglement number is concave.

\begin{thm}    
\label{thm12}
For all $\underu ,\underv\in M$, $\lambda\in\sqbrac{0,1}$ we have that
\begin{equation*}
e\sqbrac{\lambda\underu +(1-\lambda )\underv}\ge\lambda e(\underu )+(1-\lambda )e(\underv )
\end{equation*}
Moreover, if $\lambda\in (0,1)$ we have equality if and only if $\underu =\underv$.
\end{thm}
\begin{proof}
We begin with the inequality
\begin{equation*}
\doubleab{\underu}^2+\doubleab{\underv}^2\ge2\doubleab{\underu}\,\doubleab{\underv}
\end{equation*}
Hence,
\begin{align*}
1+\doubleab{\underu}^2\doubleab{\underv}^2-2\doubleab{\underu}\,\doubleab{\underv}
  &\ge 1-\doubleab{\underu}^2-\doubleab{\underv}^2+\doubleab{\underu}^2\doubleab{\underv}^2\\
  &=\paren{1-\doubleab{\underu}^2}\paren{1-\doubleab{\underv}^2}=e(\underu )^2e(\underv )^2
\end{align*}
Taking the square root gives
\begin{equation*}
1-\doubleab{\underu}\,\doubleab{\underv}\ge e(\underu )e(\underv )
\end{equation*}
It follows that
\begin{align*}
1-2\lambda (1-\lambda )\doubleab{\underu}\,\doubleab{\underv}&\ge 1-2\lambda (1-\lambda )+2\lambda (1-\lambda )e(\underu )e(\underv )\\
  &=\lambda ^2+(1-\lambda )^2+2\lambda (1-\lambda )e(\underu )e(\underv )
\end{align*}
Hence,
\begin{align*}
1-&\sqbrac{\lambda\doubleab{\underu}+(1-\lambda )\doubleab{\underv}}^2\\
  &=1-\lambda ^2\doubleab{\underu}^2-(1-\lambda )^2\doubleab{\underv}^2-2\lambda (1-\lambda )\doubleab{\underu}\,\doubleab{\underv}\\
  &\ge\lambda ^2\paren{1-\doubleab{\underu}^2}+(1-\lambda )^2\paren{1-\doubleab{\underv}^2}+2\lambda (1-\lambda )e(\underu )e(\underv )\\
  &=\sqbrac{\lambda e(\underu )+(1-\lambda )e(\underv )}^2
\end{align*}
Taking the square root and applying Schwarz's inequality gives
\begin{align*}
e\sqbrac{\lambda\underu +(1-\lambda )\underv}&=\sqbrac{1-\doubleab{\lambda\underu +(1-\lambda )\underv}^2}^{1/2}\\
  &=\sqbrac{1-\lambda ^2\doubleab{\underu}^2-(1-\lambda )^2\doubleab{\underv}^2-2\lambda (1-\lambda )\elbows{\underu ,\underv}}^{1/2}\\
  &\ge\sqbrac{1-\lambda ^2\doubleab{\underu}^2-(1-\lambda )^2\doubleab{\underv}^2-2\lambda (1-\lambda )\doubleab{\underu}\,%
  \doubleab{\underv}}^{1/2}\\
  &=\sqbrac{1-\paren{\lambda\doubleab{\underu}+(1-\lambda )\doubleab{\underv}}^2}^{1/2}\ge\lambda e(\underu )+(1-\lambda )e(\underv )
\end{align*}
If we have equality, there is equality in Schwarz's inequality. This implies that $\underu =a\underv$ for some
$a\in\real$. It
follows that $\underu =\underv$.
\end{proof}

\begin{cor}    
\label{cor13}
If $\underu ,\underv\in M$, $\lambda\in (0,1)$ and $\underu\ne\underv$ then
\begin{equation*}
e\sqbrac{\lambda\underu +(1-\lambda )\underv}>\lambda e(\underu )+(1-\lambda )e(\underv )
\end{equation*}
\end{cor}

We define $M\times M$ to be the set of probability measures on $\Natural\times\Natural$. Thus $\underu\in M\times M$ if
$\underu =\brac{u_{ij}\colon i,j\in\Natural}$, $u_{ij}\ge 0$, $\sum u_{ij}=1$. As before, the \textit{entanglement number} of $\underu$ is defined by 
\begin{equation*}
e(\underu )=\paren{1-\doubleab{\underu}^2}^{1/2}=\paren{1-\sum u_{ij}^2}^{1/2}
\end{equation*}
We also have that $e(\underu )=0$ if and only if $\underu$ is a point measure with $u_{ij}=1$ for some $i,j\in\Natural$. If
$\underv ,\underw\in M$ we define $\underu =\underv\times\underw\in M\times M$ by $u _{ij}=v_iw_j$. We say that $\underu\in M\times M$ is
\textit{factorized} if $\underu =\underv\times\underw$ for some $\underu ,\underv\in M$. If $\underu$ is not factorized, we say that
$\underu$ is \textit{entangled}. It is easy to check that $\underu$ is factorized if and only if for all $i,j\in\Natural$ we have \cite{gud19}
\begin{equation}                
\label{eq12}
u_{ij}=\sum _ju_{ij}\sum _iu_{ij}
\end{equation}
Note that if $e(\underu )=0$, then $\underu$ is factorized. However, the converse does not hold because there are factorized $\underu$ that are not point measures. For a quantum state $\psi$, we shall show that $e(\psi )=0$ if and only if $\psi$ is factorized and this will be an important difference between the quantum theory and this classical theory. It should be pointed out that $e(\psi )$ and factorization of $\psi$ are different in the quantum case, however, the analogy is similar.

\begin{exam}{2}  
(a)\enspace Let $\underu\in M\times M$ be defined by $u_{11}=1/2$, $u_{12}=1/2$. Then $e(\underu )=1/\sqrt{2\,}$ and
$\underu =\underv\times\underw$ where $v_1=1$, $w_1=1/2$, $w_2=1/2$. Thus, $\underu$ is factorized.
(b)\enspace Let $\underu\in M\times M$ be defined by $u_{11}=1/3$, $u_{12}=1/3$, $u_{22}=1/3$. Then
$\sum u_{1j}=2/3$, $\sum u_{i1}=1/3$ and $\tfrac{1}{3}\ne\tfrac{2}{9}$ so \eqref{eq12} does not hold. Hence,
$\underu$ is entangled and we have $e(\underu )=\sqrt{2/3\,}$.\hfill\qedsymbol
\end{exam}

\section{Context Coefficients}  
This section discusses the quantum statistics of operators. The basic framework for traditional quantum mechanics is a complex Hilbert space $H$. For simplicity, we shall assume that $\dim H<\infty$. Although this is a restriction, it is adequate for descriptions of quantum computation and information theory \cite{bus03,hz12,nc00}. A \textit{pure state} is represented by a one-dimensional projection $P$ on $H$. Since $P$ is one-dimensional, we can describe $P$ by a unit vector $\phi$ in its range and write $P=P_\phi=\ket{\phi}\bra{\phi}$. We also call $\phi$ a \textit{vector state} (or \textit{state vector}). A \textit{context} for a quantum system is a set of mutually orthogonal projections $P_{\phi _i}$ on $H$ such that $\sum P_{\phi _i}=I$. Equivalently, a context can be described by the corresponding orthonormal basis $\brac{\phi _i}$ of vector states. A context can be thought of as a complete set of minimal sharp events. We then see that there are an infinite uncountable number of contexts for a quantum system. This is in contrast to the classical systems described by $\Natural$ in Section~1. In that case, the minimal sharp events are just the points in $\Natural$ so the only context is $\Natural$ itself.

Let $\lscript (H)$ be the set of linear operators on $H$. The elements of $\lscript (H)$ are used to describe states, observables, symmetries and dynamics of the quantum system. If $A\in\lscript (H)$, we define the positive operator $\ab{A}$ by $\ab{A}=(A^*A)^{1/2}$. A \textit{state} is an operator $\rho\in\lscript (H)$ such that $\rho\ge 0$ and $\rmtr (\rho )=1$. Of course, a pure state is a specific type of state. We denote the set of states on $H$ by $\sscript (H)$. Any state has a spectral resolution $\rho =\sum\lambda _iP_i$ where $P_i$ are mutually orthogonal pure states, $\lambda _i\ge 0$, $\sum\lambda _i=1$. If $\rho\in\sscript (H)$ and $A\in\lscript (H)$, then the $\rho$-\textit{expectation} of $A$ is $E_\rho (A)=\rmtr (\rho A)$ and the $\rho$-\textit{variance} of $A$ is
\begin{equation*}
V_\rho (A)=E_\rho\sqbrac{\ab{A-E_\rho (A)I}^2}
\end{equation*}
In particular, for a pure state $P_\phi$ we have that
\begin{align*}
E_\phi (A)&=E_{P_\phi}(A)=\elbows{\phi ,A\phi}\\
  \intertext{and}
V_\phi (A)&=V_{P_\phi}(A)=\elbows{\phi ,\ab{A-\elbows{\phi ,A\phi}I}^2\phi}
\end{align*}
The complex vector space $\lscript (H)$ becomes a Hilbert space under the \textit{Hilbert-Schmidt inner product} $\elbows{A,B}=\rmtr (A^*B)$ \cite{hz12,nc00}. The Hilbert-Schmidt norm becomes
\begin{equation*}
\doubleab{A}=\sqbrac{\rmtr (A^*A)}^{1/2}=\sqbrac{\rmtr (\ab{A}^2)}^{1/2}
\end{equation*}

\begin{thm}    
\label{thm21}
{\rm{(a)}}\enspace $V_\rho (A)=E_\rho (\ab{A}^2)-\ab{E_\rho (A)}^2$.
{\rm{(b)}}\enspace $\ab{E_\rho (A)}^2\le E_\rho (\ab{A}^2)$ and $V_\rho (A)=0$ if and only if $A\rho ^{1/2}=c\rho ^{1/2}$ for some $c\in\complex$.
\end{thm}
\begin{proof}
The following computation proves (a).
\begin{align*}
V_\rho (A)&=\rmtr\sqbrac{\rho \paren{A-E_\rho (A)I}^*\paren{A-E_\rho (A)I}}\\
  &=\rmtr\sqbrac{\rho\paren{A^*-\overline{E_\rho (A)}I}\paren{A-E_\rho (A)I}}\\
  &=\rmtr\sqbrac{\rho\paren{\ab{A}^2-\overline{E_\rho (A)}A-E_\rho (A)A^*+\ab{E_\rho (A)}^2I}}\\
  &=E_\rho\paren{\ab{A}^2}-2\ab{E_\rho (A)}^2+\ab{E_\rho (A)}^2=E_\rho\paren{\ab{A}^2}-\ab{E_\rho (A)}^2
\end{align*}
(b)\enspace Since $V_\rho (A)\ge 0$ we have that $\ab{E_\rho (A)}^2\le E_\rho\paren{\ab{(A)}^2}$. By (a) we have
$V_\rho (A)=0$ if and only if $\ab{E_\rho (A)}^2=E_\rho\paren{\ab{A}^2}$. In terms of the Hilbert-Schmidt inner product, we have that $V_\rho (A)=0$ if and only if
\begin{align*}
\doubleab{A\rho ^{1/2}}^2&=\rmtr\sqbrac{(A\rho ^{1/2})^*A\rho ^{1/2}}=\rmtr(\rho ^{1/2}A^*A\rho ^{1/2})
  =\rmtr(\rho A^*A)\\
  &=E_\rho\paren{\ab{A}^*}=\ab{E_\rho (A)}^2=\ab{\rmtr (\rho A)}^2=\ab{\rmtr(\rho ^{1/2}A\rho ^{1/2})}^2\\
  &=\ab{\elbows{\rho ^{1/2},A\rho ^{1/2}}}^2
\end{align*}
Since $\rmtr (\rho )=1$ we have that $\doubleab{\rho ^{1/2}}=1$. Hence,
\begin{equation*}
\ab{\elbows{\rho ^{1/2},A\rho ^{1/2}}}=\doubleab{\rho ^{1/2}}\,\doubleab{A\rho ^{1/2}}
\end{equation*}
Since we have equality in Schwarz's inequality, we conclude that $A\rho ^{1/2}=c\rho ^{1/2}$ for some
$c\in\complex$.
\end{proof}

\begin{cor}    
\label{cor22}
If $\phi$ is a vector state, then
\begin{equation*}
V_\phi (A)=\elbows{\phi ,\ab{A}^2\phi}-\ab{\elbows{\phi ,A\phi}}^2
\end{equation*}
and $V_\phi (A)=0$ if and only if $A\phi =c\phi$ for some $c\in\complex$; that is, $\phi$ is an eigenvector of $A$ with eigenvalue $c$.
\end{cor}

A context given by an orthonormal basis $\ascript =\brac{\phi _i}$ can be thought of as giving a partial view of a quantum system. In order to obtain a total view we must consider various contexts \cite{gud181,gud182}. We say that $A\in\lscript (H)$ is \textit{measurable} with respect to $\ascript$ if $AP_{\phi _i}=P_{\phi _i}A$ for every $I$. In this case, $\phi _i$ is an eigenvector of $A$ with eigenvalue $\elbows{\phi _i,A\phi _i}=E_{\phi _i}(A)$. The only operators accurately described by $\ascript$ are the operators that are measurable with respect to $\ascript$ \cite{gud181,gud182}. We define the \textit{context coefficient} of $A$ with respect to $\ascript$ by
\begin{equation}                
\label{eq21}
c_\ascript (A)=\sqbrac{\sum V_{\phi _i}(A)}^{1/2}
\end{equation}

It follows from Corollary~\ref{cor22} that $c_\ascript (A)=0$ if and only if $A$ is measurable with respect to
$\ascript$. We can consider $c_\ascript (A)$ as an indicator of how close $A$ is to being measurable with respect to $\ascript$. We also see that $A$ is normal $(AA^*=A^*A)$ if and only if $c_\ascript (A)=0$ for some context
$\ascript$. For any $A\in\lscript (H)$ and context $\ascript =\brac{\phi _i}$ we can write
\begin{equation*}
A=\sum _i\elbows{\phi _i,A\phi _i}\ket{\phi _i}\bra{\phi _i}
  +\sum _{i\ne j}\elbows{\phi _i,A\phi _j}\ket{\phi _i}\bra{\phi _j}
\end{equation*}
We define the linear maps $L_\ascript ,R_\ascript\colon\lscript (H)\to\lscript (H)$ by
\begin{align*}
L_\ascript (A)&=\sum _i\elbows{\phi _i,A\phi _i}\ket{\phi _i}\bra{\phi _i}\\
R_\ascript (A)&=\sum _{i\ne j}\elbows{\phi _i,A\phi _j}\ket{\phi _i}\bra{\phi _j}
\end{align*}
and call $L_\ascript$ the \textit{context map} and $R_\ascript$ the \textit{residual map}. Thus,
$A=L_\ascript (A)+R_\ascript (A)$. Notice that $L_\ascript$ maps self-adjoint operators to self-adjoint operators, positive operators to positive operators and states to states. In fact, $L_\ascript$ is a completely positive map
\cite{gud182,hz12,nc00} and is an example of a quantum channel \cite{hz12,nc00}. Also $L_\ascript (A)$ is measurable with respect to $\ascript$ and $A$ is measurable with respect to $\ascript$ if and only if
$L_\ascript (A)=A$ or equivalently $R_\ascript (A)=0$. We remind the reader that $\doubleab{B}$ stands for the Hilbert-Schmidt norm of $B\in\lscript (H)$.

\begin{thm}    
\label{thm23}
For every $A\in\lscript (H)$ and context $\ascript =\brac{\phi _i}$ we have that
$\doubleab{R_\ascript (A)}=c_\ascript (A)$.
\end{thm}
\begin{proof}
Since
\begin{align*}
R_\ascript (A)^*R_\ascript (A)&=\sum _{i\ne j}\elbows{A\phi _j,\phi _i}\ket{\phi _j}\bra{\phi _i}
  \ctimes\sum _{r\ne s}\elbows{\phi _r,A\phi _s}\ket{\phi _r}\bra{\phi _s}\\
  &=\sum _{\scriptstyle i,j,s\atop\scriptstyle i\ne s,j}\elbows{A\phi _j,\phi _i}
  \elbows{\phi _i,A\phi _s}\ket{\phi _j}\bra{\phi _s}
\end{align*}
we conclude that
\begin{align*}
\doubleab{R_\ascript (A)}^2&=\rmtr\sqbrac{R_\ascript (A)^*R_\ascript (A)}
  =\sum _{i\ne k}\elbows{A\phi _k,\phi _i}\elbows{\phi _i,A\phi _k}\\
  &=\sum _{i,k}\elbows{A\phi _k,\phi _i}\elbows{\phi _i,A\phi _k}
  -\sum _i\elbows{A\phi _i,\phi _i}\elbows{\phi _i,A\phi _i}\\
  &=\sum _k\elbows{A\phi _k,A\phi _k}-\sum _i\ab{\elbows{\phi _i,A\phi _i}}^2\\
  &=\sum _k\paren{\elbows{\phi _k,\ab{A}^2\phi _k}-\ab{\elbows{\phi _k,A\phi _k}}^2}\\
  &=\sum _kV_{\phi _k}(A)=c_\ascript (A)^2\qedhere
\end{align*}
\end{proof}

It follows from Theorem~\ref{thm23} that $c_\ascript (A)=\doubleab{A-L_\ascript (A)}$ so that $c_\ascript (A)$ is a measure of the closeness of $A$ to $L_\ascript (A)$. Of course, $c_\ascript (A)=0$ if and only if
$A=L_\ascript (A)$, $c_\ascript (\alpha A)=\ab{\alpha}c_\ascript (A)$ and
\begin{equation*}
c_\ascript (A+B)\le c_\ascript (A)+c_\ascript (B)
\end{equation*}
Observe that $L_\ascript (A)$ is always normal with eigenvalues $\elbows{\phi _i,A\phi _i}$ and corresponding eigenvectors $\phi _i$. In general $R_\ascript (A)$ need not be normal and even when it is, its eigenstructure can be difficult to analyze except in two simple but important cases. One is when $\dim H=2$ and the other when 
$\elbows{\phi _i,A\phi _j}=\alpha$ for all $i\ne j$.

\begin{exam}{3}  
Suppose $\dim H=2$ and $R_\ascript (A)$ is normal. We can write
\begin{align*}
R_\ascript (A)&=\elbows{\phi _1,A\phi _2}\ket{\phi _1}\bra{\phi _2}
  +\elbows{\phi _2,A\phi _1}\ket{\phi _2}\bra{\phi _1}\\
  &=a\ket{\phi _1}\bra{\phi _2}+b\ket{\phi _2}\bra{\phi _1}
\end{align*}
We assume that $a,b\ne 0$ because otherwise the situation is trivial. It is easy to check that $R_\ascript (A)$ is normal if and only if $\ab{a}=\ab{b}$ in which case $a=\ab{a}e^{i\theta}$, $b=\ab{a}e^{i\phi}$, $\theta ,\phi\in\real$. Then the eigenvalues of $R_\ascript (A)$ are
\begin{equation*}
\lambda _1=\ab{a}e^{i(\theta +\phi )/2},\quad\lambda _2=-\ab{a}e^{i(\theta +\phi )/2}
\end{equation*}
with corresponding eigenvectors
\begin{equation*}
\psi _1=\tfrac{1}{\sqrt{2}}\,\sqbrac{\phi _1+e^{i(\phi -\theta )}\phi _2},\quad
  \psi _2=\tfrac{1}{\sqrt{2}}\,\sqbrac{-e^{i(\theta -\phi )/2}\phi _1+\phi _2}\rlap{$\qquad\qquad \Box$}
\end{equation*}
\end{exam}

\begin{exam}{4}  
Suppose $\dim H=n$ and $\elbows{\phi _i,A\phi _j}=\alpha\ne 0$, $i,j=1,\ldots ,n$. We then have that
\begin{equation*}
R_\ascript (A)=\alpha\sum _{i\ne j}\ket{\phi _i}\bra{\phi _j}
\end{equation*}
It follows that $R_\ascript (A)^*=\tfrac{\overalpha}{\alpha}\,R_\ascript (A)$ so $R_\alpha (A)$ is normal. For the rest of this example, we assume that $\alpha =1$ and $\alpha$ can be multiplied later if needed. First note that
$R_\ascript (A)\phi _k=\sum _{i\ne k}\phi _i$. Letting $\psi =\tfrac{1}{\sqrt{n}}\,\sum _{k=1}^n\phi _k$, it follows that
\begin{align*}
R_\ascript (A)\psi&=\tfrac{1}{\sqrt{n}}\sum _{i\ne j}\ket{\phi _i}\bra{\phi _k}\sum _{k=1}^n\phi _k
  =\tfrac{1}{\sqrt{n}}\sum _{k=1}^n\sum _{i\ne k}\phi _i\\
  &=\tfrac{n-1}{\sqrt{n}}\,\sum _{k=1}^n\psi _k=(n-1)\psi
\end{align*}
Hence, $\psi$ is a normalized eigenvector of $R_\ascript (A)$ with eigenvalue $n-1$. We will show that the other $n-1$ eigenvectors of $R_\ascript (A)$ all have eigenvalue $\minusone$ so $\minusone$ has multiplicity $n-1$. The simplest way to show this is to examine the first few cases and to observe the resulting pattern. When $n=2$, we have that $\tfrac{1}{\sqrt{2}}\,(\phi _1-\phi _2)$ is an eigenvector with eigenvalue $\minusone$. When $n=4$,
$\tfrac{1}{\sqrt{2}}\,(\phi _1-\phi _2)$, $\tfrac{1}{\sqrt{2}}\,(\phi _3-\phi _4)$,
$\tfrac{1}{2}\,(\phi _1,+\phi _2-\phi _3-\phi _4)$ are eigenvectors with eigenvalue $\minusone$. When $n=6$,
$\tfrac{1}{\sqrt{2}}\,(\phi _1-\phi _2)$, $\tfrac{1}{\sqrt{2}}\,(\phi _3-\phi _4)$, $\tfrac{1}{\sqrt{2}}\,(\phi _5-\phi _6)$,
$\tfrac{1}{2}\,(\phi _3+\phi _4-\phi _5-\phi _6)$,
$\tfrac{1}{\sqrt{8}}\,(2\phi _1+2\phi _2-\phi _3-\phi _4-\phi _5-\phi _6)$ are eigenvectors with eigenvalue
$\minusone$. When $n=3$, $\tfrac{1}{\sqrt{2}}\,(\phi _1-\phi _2)$, $\tfrac{1}{2}\,(\phi _1+\phi _2-2\phi _3)$ are eigenvectors with eigenvalue $\minusone$. When $n=5$, $\tfrac{1}{\sqrt{2}}\,(\phi _1-\phi _2)$,
$\tfrac{1}{\sqrt{2}}\,(\phi _3-\phi _4)$, $\tfrac{1}{2}\,(\phi _1+\phi _2-\phi _3-\phi _4)$,
$\tfrac{1}{\sqrt{8}}\,(\phi _1+\phi _2+\phi _3+\phi _4-4\phi _5)$ are eigenvectors with eigenvalue $\minusone$. In summary, we have the following result.\hfill\qedsymbol
\end{exam}

\begin{thm}    
\label{thm24}
Let $R_\ascript (A)=\sum _{i\ne j}\ket{\phi _i}\bra{\phi _j}$ and let $\dim H=n$. Then $R_\ascript (A)$ has eigenvalue $n-1$ with corresponding eigenvector $\psi=\tfrac{1}{\sqrt{n}}\,\sum _{k=1}^n\phi _k$ and
$R_\ascript (A)$ has eigenvalue $\minusone$ with multiplicity $n-1$ and the corresponding eigenvectors form an orthonormal basis for $\brac{\psi}^\bot$.
\end{thm}

\section{Entanglement}  
We now incorporate the two previous sections to develop a general theory of quantum entanglement. We restrict attention to bipartite systems and leave multipartite systems for later work. Let $H_1,H_2$ be finite dimensional complex Hilbert spaces and let $H=H_1\otimes H_2$. A state $\rho\in\sscript (H)$ is \textit{factorized} if there exist states $\rho _1\in\sscript (H_1)$, $\rho _2\in\sscript (H_2)$ such that $\rho =\rho _1\otimes\rho _2$. A state
$\rho\in\sscript (H)$ is \textit{separable} if $\rho$ can be written as a convex combination
$\rho =\sum\lambda _i\rho _i\otimes\sigma _i$ of factorized states. If $\rho$ is not separable, it is \textit{entangled}. Also, we say that a vector state $\psi\in H$ is \textit{factorized} if there exist vector states
$\phi _1\in H_1$, $\phi _2\in H_2$ such that $\psi =\phi _1\otimes\phi _2$. If $\psi$ is not factorized, then $\psi$ is \textit{entangled}. The following lemma summarizes some known properties of factorized states \cite{hz12}. We include the proofs for completeness.

\begin{lem}    
\label{lem31}
{\rm{(a)}}\enspace A pure state $\ket{\eta}\bra{\eta}\in\sscript (H)$ is factorized if and only if the vector state
$\eta$ is factorized
{\rm{(b)}}\enspace A pure state $\ket{\eta}\bra{\eta}\in\sscript (H)$ is separable if and only if 
$\ket{\eta}\bra{\eta}$ is factorized.
\end{lem}
\begin{proof}
(a)\enspace If $\eta\in H$ is factorized, then $\eta =\eta _1\otimes\eta _2$, $\eta _i\in H_i$, $i=1,2$. Hence,
\begin{equation*}
\ket{\eta}\bra{\eta}=\ket{\eta _1\otimes\eta _2}\bra{\eta _1\otimes\eta _2}
  =\ket{\eta _1}\bra{\eta _1}\otimes\ket{\eta _2}\bra{\eta _2}
\end{equation*}
so $\ket{\eta}\bra{\eta}$ is factorized. Conversely, if $\ket{\eta}\bra{\eta}$ is factorized, then
$\ket{\eta}\bra{\eta}=\rho _1\otimes\rho _2$, $\rho _i\in\sscript (H)$, $i=1,2$. Since
\begin{equation*}
\rho _1^2\otimes\rho _1^2=(\rho _1\otimes\rho _1)^2=\rho _1\otimes\rho _2
\end{equation*}
we have that $\rho _1^2=\rho _1$ and $\rho _2^2=\rho _2$ so $\rho _1$ and $\rho _2$ are projections. Since
\begin{equation*}
\rmtr (\rho _1)\rmtr (\rho _2)=\rmtr (\rho _1\otimes\rho _2)=1
\end{equation*}
we have that $\rmtr (\rho _1)=\rmtr (\rho _2)=1$ so $\rho _1$ and $\rho _2$ are pure states. Hence,
$\rho _1=\ket{\phi _1}\bra{\phi _1}$, $\rho _2=\ket{\phi _2}\bra{\phi _2}$, $\phi _i\in H_i$, $i=1,2$ and we have that
\begin{equation*}
\ket{\eta}\bra{\eta}=\ket{\phi _1}\bra{\phi _1}\otimes\ket{\phi _2}\bra{\phi _2}
  =\ket{\phi _1\otimes\phi _2}\bra{\phi _1\otimes\phi _2}
\end{equation*}
Thus, $\eta =\phi _1\otimes\phi _2$ so $\eta$ is factorized.\newline
(b)\enspace If $\ket{\eta}\bra{\eta}$ is factorized, then clearly $\ket{\eta}\bra{\eta}$ is separable. Conversely, if $\ket{\eta}\bra{\eta}$ is separable, then there are $\lambda _i>0$ with
\begin{equation*}
\ket{\eta}\bra{\eta}=\sum _{i=1}^n\lambda _i\rho _i\otimes\sigma _i
\end{equation*}
We concluded that $\lambda _i\rho _i\otimes\sigma _i\le\ket{\eta}\bra{\eta}$ and since
$\ket{\eta}\bra{\eta}$ is one-dimensional we have that
$\lambda _i\rho _i\otimes\sigma _i=\lambda\ket{\eta}\bra{\eta}$ for some $\lambda\in\sqbrac{0,1}$. Taking the trace gives $\lambda _i=\lambda$ so $\ket{\eta}\bra{\eta}=\rho _i\otimes\sigma _i$, $i=1,\ldots ,n$. Therefore, $\ket{\eta}\bra{\eta}=\rho _1\otimes\sigma _1$ so $\ket{\eta}\bra{\eta}$ is factorized.
\end{proof}

Let $\ascript =\brac{\phi _i}$, $\bscript =\brac{\psi _i}$ be orthonormal bases (contexts) for $H_1$ and $H_2$, respectively. If $\underlambda\in M$ is a probability measure, we call $(\underlambda ,\ascript ,\bscript )$ an
\textit{entanglement} and we call $(M,\ascript ,\bscript )$ an \textit{entanglement system}. We assume without loss of generality that $\dim H_1=\dim H_2=n$. We can do this because if $\dim H_1<\dim H_2$, say, then we can enlarge $H_1$ to $\dim H_2$ and no harm is done. Moreover, we assume that
$\rmsupp (\underlambda )\subseteq\brac{1,2,\ldots ,n}$. Corresponding to an entanglement $E=(\underlambda ,\ascript ,\bscript )$ we have a vector state
\begin{equation*}
\psi _E=\sum\sqrt{\lambda _i\,}\phi _i\otimes\psi _i\in H_1\otimes H_2
\end{equation*}
a pure state $P_E=P_{\psi _E}$, a separable state
\begin{equation*}
\rho _E=\sum\lambda _iP_{\phi _i\otimes\psi _i}=\sum\lambda _iP_{\phi _i}\otimes P_{\psi _i}
\end{equation*}
and an \textit{entanglement operator}
\begin{align*}
B_E&=\sum _{i\ne j}\sqrt{\lambda _i\lambda _j\,}\,\ket{\phi _i\otimes\psi _i}\bra{\phi _j\otimes\psi _j}\\
  &=\sum _{i\ne j}\sqrt{\lambda _i\lambda _j\,}\,\ket{\phi _i}\bra{\phi _j}\otimes\ket{\psi _i}\bra{\psi _j}
\end{align*}
From Section~1, since $\underlambda\in M$ we have the entanglement number $e(\underlambda )$. We use this to define the \textit{entanglement number}
\begin{equation*}
e(\psi _E)=e(P_E)=e(\underlambda )
\end{equation*}

Conversely, if $\psi\in H_1\otimes H_2$ is a vector state, then there exists a \textit{Schmidt decomposition} consisting of an entanglement $(\underlambda ,\ascript ,\bscript )$ where $\underlambda\in M$ is unique and
$\psi =\sum\sqrt{\lambda _i\,}\phi _i\otimes\psi _i$ \cite{cla05,hz12,nc00}. In this way, any vector state $\psi$ determines an entanglement $E=(\underlambda ,\ascript ,\bscript )$ so that $\psi =\psi _E$ although $\ascript$ and $\bscript$ need not be unique. It is easy to check that
\begin{equation*}
P_E=\ket{\psi _E}\bra{\psi _E}=\rho _E+B_E
\end{equation*}
and $B_E$ is a self-adjoint, traceless operator. We consider $\rho _E$ as the non-entangled part of $P_E$ and $B_E$ as describing the entangled part. Letting $\dscript =\ascript\otimes\bscript =\brac{\phi _i\otimes\psi _j}$ be the corresponding orthonormal basis (context) for $H=H_1\otimes H_2$ we have that $\rho _E=L_\dscript (P_E)$ and $B_E=R_\dscript (P_E)$ where $L_\dscript$ and $R_\dscript$ are the context map and residual map of Section~2.

Considering the Hilbert-Schmidt norm $\doubleab{B_E}$ we see that $\doubleab{B_E}=\doubleab{P_E-\rho _E}$ gives a measure of the entanglement of $P_E$. Thus, if $\doubleab{B_E}$ is small, then $P_E$ is close to
$\rho _E$ and is less entangled and when $\doubleab{B_E}$ is large, then $P_E$ is more entangled. The next result shows that our three entanglement measures coincide.

\begin{thm}    
\label{thm32}
$c_\dscript (B_E)=\doubleab{B_E}=e(\psi _E)$
\end{thm}
\begin{proof}
It follows from Theorem~\ref{thm23} that $c_\dscript (B_E)=\doubleab{B_E}$. To show that
$\doubleab{B_E}=e(\psi _E)$ we have that
\begin{align*}
B_E^2&=\sqbrac{\sum _{i\ne j}\sqrt{\lambda _i\lambda _j\,}\,\ket{\phi _i\otimes\psi _i}\bra{\phi _j\otimes\psi _j}}
  \sqbrac{\sum _{r\ne s}\sqrt{\lambda _r\lambda _s\,}\,\ket{\phi _r\otimes\psi _r}\bra{\phi _s\otimes\psi _s}}\\
  &=\sum _{i\ne j}\sum _s\sqrt{\lambda _i\lambda _s\,}\ket{\phi _i\otimes\psi _i}\bra{\phi _s\otimes\psi _s}\\
  &=\sum _{i,s}(1-\lambda _i)\sqrt{\lambda _i\lambda _s\,}\,\ket{\phi _i\otimes\psi _i}\bra{\phi _s\otimes\psi _s}
\end{align*}
Hence,
\begin{equation*}
\rmtr (B_E^2)=\sum (1-\lambda _i)\lambda _i=1-\sum\lambda _i^2=1-\doubleab{\underlambda}^2
\end{equation*}
We conclude that
\begin{equation*}
\doubleab{B_E}=\sqbrac{\rmtr (B_E^2)}^{1/2}=\sqrt{1-\doubleab{\underlambda}^2\,}=e(\underlambda )
  =e(\psi _E)\qedhere
\end{equation*}
\end{proof}

Let $E=(\underalpha ,\ascript ,\bscript )$ and $F=(\underbeta ,\ascript ,\bscript )$ be entanglements belonging to the same entanglement system $(M,\ascript ,\bscript )$. We have the corresponding vector states
$\psi _E=\sum\sqrt{\alpha _i\,}\phi _i\otimes\psi _i$, $\psi _F=\sum\sqrt{\beta _i\,}\phi _i\otimes\psi _i$. For
$\lambda\in (0,1)$ we have the entanglement
\begin{equation*}
G=\paren{\lambda\underalpha +(1-\lambda )\underbeta ,\ascript ,\bscript}
\end{equation*}
and vector state
\begin{equation*}
\psi _G=\sum\sqrt{\lambda\alpha _i+(1-\lambda )\beta _i\,}\,\,\phi _i\otimes\psi _i
\end{equation*}
By Theorem~\ref{thm12} we have that
\begin{equation*}
e(\psi _G)=e\sqbrac{\lambda\underalpha +(1-\lambda )\underbeta}
  \ge\lambda e(\underalpha )+(1-\lambda )e(\underbeta )
  =\lambda e(\psi _E)+(1-\lambda )e(\psi _F)
\end{equation*}

Our entanglement number is related to entanglement robustness \cite{bcp18,cla05,pv07,ste03,vp98,vt99,wg03},
but there are important differences and the motivation is not the same. We leave a detailed comparison to later work.

\begin{exam}{5}  
Let $H=\complex ^2\otimes\complex ^2$ and let $\psi\in H$ be a vector state. By the Schmidt decomposition, there are numbers $\lambda _1,\lambda _2\ge 0$ with $\lambda _1+\lambda _2=1$ and bases
$\ascript =\brac{\phi _1,\phi _2}$, $\bscript =\brac{\psi _1,\psi _2}$ of $\complex ^2$ such that
\begin{equation*}
\psi =\sqrt{\lambda _1\,}\,\phi _1\otimes\psi _1+\sqrt{\lambda _2\,}\,\phi _2\otimes\psi _2
\end{equation*}
We have that $P_\psi =\rho _\psi +B_\psi$ where
\begin{align*}
\rho _\psi&=\lambda _1\ket{\phi _1\otimes\psi _1}\bra{\phi _1\otimes\psi _1}
  +\lambda _2\ket{\phi _2\otimes\psi _2}\bra{\phi _2\otimes\psi _2}\\
  \intertext{and}
B_\psi&=\sqrt{\lambda _1\lambda _2\,}\,\sqbrac{\ket{\phi _1\otimes\psi _1}\bra{\phi _2\otimes\psi _2}%
  +\ket{\phi _2\otimes\psi _2}\bra{\phi _1\otimes\psi _1}}
\end{align*}
We see that $\rho _\psi$ is a separable state and the entanglement operator $B_\psi$ is self-adjoint and traceless. The eigenvalues of $B_\psi$ are $0,0,\sqrt{\lambda _1\lambda _2\,},-\sqrt{\lambda _1\lambda _2\,}$. The corresponding eigenvectors are $\phi _1\otimes\psi _2$, $\phi _2\otimes\psi _1$ which are factorized and
\begin{equation*}
\tfrac{1}{\sqrt{2\,}}(\phi _1\otimes\psi _1+\phi _2\otimes\psi _2),\quad
  \tfrac{1}{\sqrt{2\,}}(\phi _1\otimes\psi _1-\phi _2\otimes\psi _2)
\end{equation*}
which are entangled. The Hilbert-Schmidt norm of $B_\psi$ is
\begin{equation*}
\doubleab{B_\psi}=\sqrt{2\lambda _1\lambda _2\,}=e(\psi )\rlap{$\qquad\qquad \Box$}
\end{equation*}
\end{exam}

\begin{exam}{6}  
Let $E=(M,\ascript ,\bscript )$ be an entanglement system with $\ascript =\brac{\phi _i}$, $\bscript =\brac{\psi _j}$. Consider the following vector states in $E$
\begin{align*}
\alpha&=\tfrac{1}{\sqrt{2\,}}\,\phi _1\otimes\psi _1+\tfrac{1}{\sqrt{2\,}}\,\phi _2\otimes\psi _2\\
\beta&=\tfrac{1}{\sqrt{3\,}}\,\phi _1\otimes\psi _1+\tfrac{1}{\sqrt{3\,}}\,\phi _2\otimes\psi _2
  +\tfrac{1}{\sqrt{3\,}}\,\phi _3\otimes\psi _3\\
  \gamma&=\tfrac{1}{\sqrt{2\,}}\,\phi _1\otimes\psi _1+\tfrac{1}{\sqrt{3\,}}\,\phi _2\otimes\psi _2
  +\tfrac{1}{\sqrt{6\,}}\,\phi _3\otimes\psi _3\\
  \delta&=\tfrac{1}{3}\,\phi _1\otimes\psi _1+\tfrac{1}{3}\,\phi _2\otimes\psi _2
  +\sqrt{\tfrac{7}{9}\,}\,\phi _3\times\psi _3
\end{align*}
All of these states are entangled and as in Example~1 we have $e(\alpha )=1/\sqrt{2}$, $e(\beta )=\sqrt{2/3\,}$,
$e(\gamma )=\sqrt{11/18\,}$, $e(\delta )=\sqrt{30\,}/9$ and we have that
\begin{equation*}
e(\delta )<e(\alpha )<e(\gamma )<e(\beta )\rlap{$\qquad\qquad \Box$}
\end{equation*}
\end{exam}

\begin{exam}{7}  
If $\ascript =\brac{\phi _i}$ is an orthonormal basis for $H$, the corresponding \textit{symmetric-antisymmetric} basis for $H\otimes H$ is
\begin{equation*}
\ascript ^{SA}=\brac{\phi _i\otimes\phi _i, \tfrac{1}{\sqrt{2\,}}\,(\phi _i\otimes\phi _j+\phi _j\otimes\phi _i),%
  \tfrac{1}{\sqrt{2\,}}(\phi _i\otimes\phi _j-\phi _j\otimes\phi _i), i<j}  
\end{equation*}
The first two types are symmetric and the last type are antisymmetric. There are $n(n+1)/2$ symmetric and $n(n-1)/2$ antisymmetric states. The entanglement number for the first type is $0$ and the others are
$1/\sqrt{2\,}$. We have that $P_{\phi _i\otimes\phi _i}=P_{\phi _i}\otimes P_{\phi _i}$ is factorized and
\begin{align*}
P_{\tfrac{1}{\sqrt{2\,}}\,(\phi _i\otimes\phi _j+\phi _j\otimes\phi _i)}
  &=\tfrac{1}{2}\,\ket{\phi _i}\bra{\phi _i}\otimes\ket{\phi _j}\bra{\phi _j}
  +\tfrac{1}{2}\,\ket{\phi _j}\bra{\phi _j}\otimes\ket{\phi _i}\bra{\phi _i}\\
  &\qquad +\tfrac{1}{2}\,\ket{\phi _j}\bra{\phi _i}\otimes\ket{\phi _i}\bra{\phi _j}
  +\tfrac{1}{2}\,\ket{\phi _i}\bra{\phi _j}\otimes\ket{\phi _j}\bra{\phi _i}\\
  &=\tfrac{1}{2}\,P_{\phi _i}\otimes P_{\phi _j}+\tfrac{1}{2}\,P_{\phi _j}\otimes P_{\phi _i}
  +\rmre\paren{\ket{\phi _j}\bra{\phi _i}\otimes\ket{\phi _i}\bra{\phi _j}}
\end{align*}

We can write this as $A+B$ where $A$ is the separable state
\begin{equation*}
A=\tfrac{1}{2}\,P_{\phi _i}\otimes P_{\phi _j}+\tfrac{1}{2}\,P_{\phi _j}\otimes P_{\phi _i}
\end{equation*}
and $B$ is the entanglement operator. We also have
\begin{equation*}
P_{\tfrac{1}{\sqrt{2\,}}\,(\phi _i\otimes\phi _j-\phi _j\otimes\phi _i)}=A-B\rlap{$\qquad\qquad \Box$}
\end{equation*}
\end{exam}

\begin{exam}{8}  
Let $H=H_1\otimes H_2$ with $\dim H_1=\dim H_2=n$ and let $\psi\in H$ be the maximally entangled vector given by
\begin{equation*}
\psi =\tfrac{1}{\sqrt{n\,}}\,\sum\phi _i\otimes\psi _i
\end{equation*}
where $\ascript =\brac{\phi _i}$, $\bscript =\brac{\psi _i}$ are orthonormal bases for $H_1,H_2$, respectively. Letting $\lambda _i=1/n$, $i=1,2,\ldots ,n$ and $E=(\underlambda ,\ascript ,\bscript )$ we have that $\psi =\psi _E$ with corresponding pure state $P_E$ and entanglement operator
\begin{equation*}
B_E=\tfrac{1}{n}\,\sum\ket{\phi _i}\bra{\phi _j}\otimes\ket{\psi _i}\bra{\psi _j}
\end{equation*}
Letting $\dscript =\brac{\phi _i\otimes\psi _j}$ be the resulting orthonormal basis for $H$ and $R_\dscript$ be the corresponding residual map we have $B_E=R_\dscript (P_E)$. It follows from Theorem~\ref{thm24} that the nonzero eigenvalues of $B_E$ are $1-\tfrac{1}{n}$ and $-\tfrac{1}{n}$. The eigenvalue $1-\tfrac{1}{n}$ has multiplicity $1$ and corresponding eigenvector $\psi$ while the eigenvalue $-\tfrac{1}{n}$ has multiplicity $n-1$ whose eigenspace is the subspace of $H$ generated by $\brac{\phi _i\otimes\psi _j\colon i\ne j}$ and orthogonal to $\phi$.\hfill\qedsymbol
\end{exam}

Until now we have considered the entanglement number for a pure state $P_\phi$. For the remainder of this article we shall discuss mixed states. If $\rho$ is a mixed state on $H$ that is not pure, then $\rho$ possesses an uncountably infinite number of decompositions $\rho =\sum\lambda _iP_i$, $\lambda _i>0$, $\sum\lambda _i=1$ where $P_i$ are pure states \cite{hz12}. Also, $\rho$ has a \textit{spectral decomposition} $\rho=\sum\mu _iQ_i$, $\mu _i>0$, $\sum\mu _i=1$, where $Q_i$ are mutually orthogonal pure states. The $\mu _i$ are the nonzero eigenvalues of $\rho$ and the ranges of $Q_i$ are the corresponding eigenvectors of $\rho$. The next example is based on Example~6.13 \cite{hz12}.

\begin{exam}{9}  
Let $H=\complex ^2\otimes\complex ^2$, let $\brac{\phi _1,\phi _2}$ be an orthonormal basis for $\complex ^2$ and define $\phi =\tfrac{1}{\sqrt{2\,}}(\phi _1+\phi _2)$. We now consider the separable state
\begin{equation*}
\rho =\tfrac{1}{2}\,\paren{\ket{\phi\otimes\phi}\bra{\phi\otimes\phi}%
  +\ket{\phi _1\otimes\phi _1}\bra{\phi _1\otimes\phi _1}}
\end{equation*}
The eigenvalues of $\rho$ are $0$ (multiplicity 2), $1/4$ and $3/4$. The eigenvectors for $0$ are
\begin{equation*}
\psi _1=\tfrac{1}{\sqrt{2\,}}(\phi _2-\phi _1)\otimes\phi _2,\qquad
  \psi _2=\tfrac{1}{\sqrt{6\,}}\sqbrac{(\phi _1+\phi _2)\otimes\phi _2-2\phi _2\otimes\phi _1}
\end{equation*}
The eigenvectors for $1/4$ and $3/4$ are
\begin{align*}
\psi _3&=\tfrac{1}{2\sqrt{3\,}}\sqbrac{(3\phi _1+\phi _2)\otimes\phi _1+(\phi _1+\phi _2)\otimes\phi _2}\\
\psi _4&=\tfrac{1}{2}\,\sqbrac{(\phi _2-\phi _1)\otimes\phi _1+(\phi _1+\phi _2)\otimes\phi _2}
\end{align*}
The unique spectral decomposition of $\rho$ becomes
\begin{equation}                
\label{eq31}
\rho =\tfrac{1}{4}\,P_{\psi _3}+\tfrac{3}{4}\,P_{\psi _4}
\end{equation}
Notice that $\psi _3$ and $\psi _4$ are entangled. This gives an example of a separable state whose unique spectral decomposition consists of entangled pure states.\hfill\qedsymbol
\end{exam}

Example~9 shows that a spectral decomposition cannot be used to determine an entanglement number for a mixed state. Indeed, in \eqref{eq31} since $\rho$ is separable the entanglement number for $\rho$ should be zero, yet the entanglement number for $P_{\psi _3}$ and $P_{\psi _4}$ are positive.

We now define the entanglement number for a mixed state $\rho$. Suppose $\rho =\sum\lambda _iP_i$,
$\lambda _i>0$, $\sum\lambda _i=1$ is a decomposition of $\rho$ into pure states $P_i$, where $P_i\ne P_j$,
$i\ne j$. Let $\ascript =\brac{P_i}$ and define
\begin{equation*}
e_\ascript (\rho )=\sum\lambda _ie(P_i)
\end{equation*}
We define the \textit{entanglement number} $e(\rho )$ by
\begin{equation}                
\label{eq32}
e(\rho )=\inf _\ascript\sqbrac{e_\ascript (\rho )}
\end{equation}
Since a pure state has the decomposition $P=P$, \eqref{eq32} reduces to the usual definition of entanglement number for pure states. We say that the infimum is \eqref{eq32} is \textit{attained} if there is an $\ascript$ such that $e(\rho )=e_\ascript (\rho )$. It is an open problem whether the infimum is always attained.

\begin{thm}    
\label{thm33}
A state $\rho$ is separable if and only if $e(\rho )$ is attained and $e(\rho )=0$.
\end{thm}
\begin{proof}
If $\rho$ is separable we have that $\rho =\sum\lambda _iP_i$ where $P_i$ are factorized pure states. Since $e(P_i)=0$ for all $i$, we have that $e_\ascript (\rho )=0$ for $\ascript =\brac{P_i}$. Hence, $e(\rho )=0$.
Conversely, suppose $e(\rho )$ is attained at $\ascript =\brac{P_i}$ and $e(\rho )=e_\ascript (\rho )=0$.
Since $\rho =\sum\lambda _iP_i$, $\lambda _i>0$, $\sum\lambda _i=1$ and
\begin{equation*}
\sum\lambda _ie(P_i)=e(\rho )=0
\end{equation*}
we conclude that $e(P_i)=0$ for all $i$. It follows that $P_i$ is factorized for all $i$ and hence $\rho$ is separable.
\end{proof}

It follows that if $\rho$ is separable, then $e(\rho )=0$ and if $e(\rho )>0$ or is not attained, then $\rho$ is entangled.

\end{document}